\documentclass[preprint]{elsarticle}

\usepackage{array,xspace,multirow,hhline,tikz,colortbl,tabularx,booktabs,fixltx2e,amsmath,amssymb,amsfonts,amsthm}
\usepackage{algorithm}
\usepackage{algorithmic}
\usepackage{verbatim,ifthen}
\usepackage{enumitem}
\usepackage{pifont}
\usepackage{ifthen}
\usepackage{calrsfs,mathrsfs}
\usepackage{bbding,pifont}
\usepackage{pgflibraryshapes}
\usepackage{nicefrac}
\usepackage{url}
\definecolor{light-gray}{gray}{0.9}

	\newtheorem{theorem}{Theorem}%

    \newtheorem*{theorem*}{Theorem}
        \newtheorem{definition}{Definition}

		\newcommand{\ml}[1][]{\ifthenelse{\equal{#1}{}}{\mathit{ML}}{\mathit{ML}(#1)}}
		\newcommand{\sml}[1][]{\ifthenelse{\equal{#1}{}}{\mathit{SML}}{\mathit{SML}(#1)}}
		\newcommand{\sd}[1][]{\ifthenelse{\equal{#1}{}}{\mathit{SD}}{\mathit{SD}(#1)}}
		\newcommand{\rsd}[1][]{\ifthenelse{\equal{#1}{}}{\mathit{RSD}}{\mathit{RSD}(#1)}}
		\newcommand{\st}[1][]{\ifthenelse{\equal{#1}{}}{\mathit{ST}}{\mathit{ST}(#1)}}
		\newcommand{\bd}[1][]{\ifthenelse{\equal{#1}{}}{\mathit{BD}}{\mathit{BD}(#1)}}
		\newcommand{\pc}[1][]{\ifthenelse{\equal{#1}{}}{\mathit{PC}}{\mathit{PC}(#1)}}
		\newcommand{\dl}[1][]{\ifthenelse{\equal{#1}{}}{\mathit{DL}}{\mathit{DL}(#1)}}
		\newcommand{\ul}[1][]{\ifthenelse{\equal{#1}{}}{\mathit{UL}}{\mathit{UL}(#1)}}

		\newcommand{\serdict}[1][]{\ifthenelse{\equal{#1}{}}{\sigma}{\sigma(#1)}}

	\newcommand\eat[1]{}

	\usepackage{enumitem}
	\setenumerate[1]{label=\rm(\it{\roman{*}}\rm),ref=({\it\roman{*}}),leftmargin=*}
	\newlength{\wordlength}

	\newcommand{\midd}{\mathbin{:}}

	\newcommand{\eqclass}[2][]{\ifthenelse{\equal{#1}{}}{[#2]}{[#2]_{\sim_{#1}}}}

	\newcommand{\Pref}[1][]{
		\ifthenelse{\equal{#1}{}}{\mathrel R}{\mathop{R_{#1}}}
	}                                          
	\newcommand{\sPref}[1][]{                  
		\ifthenelse{\equal{#1}{}}{\mathrel P}{\mathop{P_{#1}}}
	}                                          
	\newcommand{\Indiff}[1][]{                 
		\ifthenelse{\equal{#1}{}}{\mathrel I}{\mathop{I_{#1}}}
	}
	\newcommand{\prefset}[1][]{\ifthenelse{\equal{#1}{}}{\mathcal{R}}{\mathcal{R}_{#1}}}

\definecolor{PurplePlum}{rgb}{0.1,0,0.55} 
\definecolor{Brown}{rgb}{0.5,.25,0}
\definecolor{Green}{rgb}{0,.5,0}
\definecolor{Orange}{rgb}{1,.5,0}
\definecolor{Gray}{rgb}{0.5,0.5,0.5}
\definecolor{Black}{rgb}{0,0,0}

\usepackage{color} 
\usepackage[normalem]{ulem} 

\newcommand{\jr}{\ensuremath{\mathit{JR}}\xspace}
\newcommand{\ejr}{\ensuremath{\mathit{EJR}}\xspace}
\newcommand{\pjr}{\ensuremath{\mathit{PJR}}\xspace}

\newcommand{\calA}{{\vec{A}}}

\begin{document}
	\title{Computational Complexity of Testing \\Proportional Justified Representation
	}

	 \author{Haris Aziz}
	  \ead{haris.aziz@data61.csiro.au}
      \address{Data61, CSIRO and UNSW Australia\\
      Computer Science and Engineering, 
      Building K17, UNSW, Sydney NSW 2052, Australia} 
      
		\author{Shenwei Huang} \ead{shenwei.huang@unsw.edu.au}
      \address{UNSW Australia\\
      Computer Science and Engineering, 
      Building K17, UNSW, Sydney NSW 2052, Australia} 
	


	\begin{abstract}
		We consider a committee voting setting in which each voter approves of a subset of candidates and based on the approvals, a target number of candidates are selected. Aziz et al. (2015) proposed two  representation axioms called justified representation and extended justified representation. Whereas the former can be tested as well as achieved in polynomial time, the latter property is coNP-complete to test and no polynomial-time algorithm is known to achieve it. Interestingly, S{\'a}nchez-Fern{\'a}ndez et~al. (2016) proposed an intermediate property called proportional justified representation that admits a polynomial-time algorithm to achieve. The complexity of testing proportional justified representation has remained an open problem. In this paper, we settle the complexity by proving that testing proportional justified representation is coNP-complete. 
	We complement the complexity result by showing that the problem admits efficient algorithms if any of the following parameters are bounded: (1) number of voters (2) number of candidates (3) maximum number of candidates approved by a voter (4) maximum number of voters approving a given candidate. 
	\end{abstract}



\maketitle

\newpage

\section{Introduction}


We consider a committee voting setting in which each voter approves of a subset of candidates and based on the approvals, a target $k$ number of candidates are selected. The setting has been referred to as approval-based multiwinner voting or committee voting with approvals. The setting has inspired a number of natural voting rules~\citep{Kilg10a,BrFi07c,LMM07a,AGG+15a,SFL16a}. Many of the voting rules attempt to satisfy some notion of representation. However it has been far from clear what axiom captures the representation requirements.

\citet{ABC+15a,ABC+16a} proposed two compelling representation axioms called \emph{justified representation (\jr)} and \emph{extended justified representation (\ejr)}. Whereas the former can be tested as well as achieved in polynomial time, the latter property is coNP-complete to test and no polynomial-time algorithm is known to achieve it. Interestingly, \citet{SFFB16a} proposed an intermediate property called \emph{proportional justified representation (\pjr)} that admits a polynomial-time algorithm to achieve~\citep{BFJL16a,SFF16a}.\footnote{The property \pjr was independently proposed by Haris Aziz in October 2014 who referred to it as weak \ejr.} 
The idea behind all the three properties is that a cohesive and large enough group deserves sufficient number of approved candidates in the winning set of candidates.
 \citet{SFFB16a} argued that although \ejr is a stronger property than \pjr, \pjr is more reasonable because it is compatible with a property called perfect representation.

Proportional justified representation (\pjr) has been examined in subsequent  papers~\citep{BFJL16a,SFF16a,SFF+17a}. 
Despite the flurry of work on the property, the complexity of testing proportional justified representation has remained an open problem. 
\citet{SFF+17a} state that \emph{``we do not know what is the complexity of checking whether a given committee provides PJR''. }
In a talk ``Approval Voting, representation, \& Liquid Democracy'' at the\emph{ Workshop on Future Directions in Computational Social Choice, Hungary} in November 2016, Markus Brill also mentioned the problem as an interesting open problem.
\footnote{\url{http://econ.core.hu/file/download//future_markus.pdf}}
The problem is especially important if one wants to test whether a status quo outcome or the outcome of some other rule or negotiation process satisfies \pjr.
Previously, \citet{ALL16a} studied the complexity of testing Pareto optimality of a committee.

In this paper, we settle the complexity of testing \pjr by proving that the problem is coNP-complete.
We complement the complexity result by showing that the problem admits efficient algorithms if any of the following parameters are bounded: (1) $n$ (number of voters) (2) $m$ (number of candidates) (3) $a$ (maximum number of candidates approved by a voter) (4) $d$ (maximum number of voters approving a given candidate). 
For the first two parameters, we show that the problem is  \emph{FPT (fixed-parameter tractable)}, i.e, there exists an FPT algorithm that
solves the problem in $f(k) \cdot poly(|I|)$ time, where
where $k$ is the parameter and $f$ is some computable
function and $poly$ is a polynomial both independent of problem instance $I$. 

Our results are summarized in Table~\ref{table:summary:pjr}.

\begin{table}[h!]
	\centering
	\scalebox{0.95}{
	\begin{tabular}{lll}
		\toprule
		Parameter & Complexity & Reference \\
		\midrule
		--- & coNP-complete
		  & Th.~\ref{th:pjr-hard}\\
		$n$: \# voters & in FPT: $O(2^{n}{mn})$
		  & Th.~\ref{th:fpt-n}\\
		 $m$: \# candidates & in FPT: $O(2^{m}{m}^3n)$
		  & Th.~\ref{th:fpt-m}\\
 		 $a:\max_{i\in N}|A_i|$ &in P for constant $a$: $O(m^{a+1}{m}^2n)$&Th.~\ref{th:p-a}\\
 		 $d: \max_{c\in C}|\{i\in N\midd c\in A_i\}|$ &in P for constant $d$: $O(n^ddnm^2)$ &Th.~\ref{th:p-d}\\
		\bottomrule
	\end{tabular}
	}
\centering
	\caption{Complexity of testing \pjr}
	\label{table:summary:pjr}
\end{table}

%
%
%


%
%
%
%

%
%
%
%


\section{Approval-based Committee Voting and Representation Properties}

We consider a social choice setting with a set $N=\{1,\ldots, n\}$ of voters and a set $C$ of $m$ candidates. 
Each voter $i\in N$ submits an approval ballot $A_i\subseteq C$, which represents the subset of candidates that she
approves of. We refer to the list $\calA = (A_1,\ldots, A_n)$ of approval ballots as the {\em ballot profile}. 
We will consider {\em approval-based multi-winner voting rules} that take as input a tuple $(N, C, \calA, k)$, 
where $k$ is a positive integer that satisfies $k\le m$, and return a subset 
$W \subseteq C$ of size $k$, which we call the {\em winning set}, or {\em committee}. 

We now summarize the main representation properties proposed in the literature.



\begin{definition}[Justified representation (JR)]
Given a ballot profile $\calA = (A_1, \dots, A_n)$ over a candidate set $C$ and a target committee size $k$,
we say that a set of candidates $W$ of size $|W|=k$ {\em satisfies justified representation 
for $(\calA, k)$} if 
\[\forall X\subseteq N: |X|\geq \frac{n}{k} \text{ and } |\cap_{i\in X}A_i|\geq 1 \implies (|W\cap (\cup_{i\in X}A_i)|\geq 1)\]

\end{definition}

The rationale behind this definition is that if $k$ candidates are to be selected, then, intuitively,
each group of $\frac{n}{k}$ voters ``deserves'' a representative. Therefore, a set of $\frac{n}{k}$ voters 
that have at least one candidate in common should not be completely unrepresented.


\begin{definition}[Extended justified representation (EJR)]
Given a ballot profile $(A_1, \dots, A_n)$ over a candidate set $C$, a target committee size $k$, $k\le m$,
we say that a set of candidates $W$, $|W|=k$, {\em satisfies  $\ell$-extended justified representation
for $(\calA, k)$}  and integer $\ell$ if
\[\forall X\subseteq N: |X|\geq \ell\frac{n}{k} \text{ and } |\cap_{i\in X}A_i|\geq \ell \implies (\exists i\in X: |W\cap A_i|\geq \ell).\]

We say that $W$ {\em satisfies extended justified representation for $(\calA, k)$} if it satisfies {satisfies $\ell$-extended justified representation
for $(\calA, k)$} and all integers $\ell\leq k$.

\end{definition}

\citet{SFFB16a} came up with the notion of \emph{proportional 
justified representation (\pjr)}, which can be seen as an alternative to \ejr.

\begin{definition}[Proportional Justified Representation (\pjr)]

Given a ballot profile $(A_1, \dots, A_n)$ over a candidate set $C$, a target committee size $k$, $k\le m$, and integer $\ell$
we say that a set of candidates $W$, $|W|=k$, {\em satisfies $\ell$-proportional justified representation
for $(\calA, k)$}  if
\[\forall X\subseteq N: |X|\geq \ell\frac{n}{k} \text{ and } |\cap_{i\in X}A_i|\geq \ell \implies (|W\cap (\cup_{i\in X}A_i)|\geq \ell)\]

We say that $W$ {\em satisfies proportional justified representation for $(\calA, k)$} if it satisfies {satisfies $\ell$-proportional justified representation
for $(\calA, k)$} and all integers $\ell\leq k$.
\end{definition}

It is easy to observe that EJR implies PJR which implies JR.

		\section{Results}

		We first prove that testing \pjr is coNP-complete. The proof involves a similar type of reduction as the one used by \citet{ABC+15a,ABC+16a} to prove that testing \ejr is coNP-complete.

		\begin{theorem}\label{th:pjr-hard}
		Given a ballot profile $\calA$, a target committee size $k$, and a committee $W$, $|W|=k$,
		it is {\em coNP}-complete to check whether $W$ satisfies \pjr for $(\calA, k)$.
		\end{theorem}
		\begin{proof}
		It is easy to see that this problem is in coNP. A set of voters $X\subset N$ such that
$|X|\geq \ell\frac{n}{k}$,  $|\cap_{i\in X}A_i|\geq \ell$ and $|W\cap (\cup_{i\in X}A_i)|\geq \ell)$ is a certificate that $W$ does not satisfy \pjr.

		To prove coNP-completeness, we reduce the classic {\sc Balanced Biclique} problem
		([GT24] in \citeauthor{GaJo79a} \citeyear{GaJo79a})
		to the complement of our problem. An instance of {\sc Balanced Biclique} is given by a bipartite
		graph $(L, R, E)$ with parts $L$ and $R$ and edge set $E$, and an integer $\ell$; it is a ``yes''-instance
		if we can pick subsets of vertices $L'\subseteq L$ and $R'\subseteq R$ so that $|L'|=|R'|=\ell$
		and $(u, v)\in E$ for each $u\in L', v\in R'$; otherwise, it is a ``no''-instance.

		Given an instance $\langle (L, R, E), \ell\rangle$ of {\sc Balanced Biclique} with $R=\{v_1, \dots, v_s\}$, 
		we create an instance of our problem as follows. Assume without loss of generality that $s\ge 3$, $\ell\ge 3$.
		We construct $3$ pairwise disjoint sets of candidates $C_0$, $C_1$ and $C_2$, 
		so that $C_0=L$, $|C_1|=\ell-1$, $|C_2|=s\ell+\ell-3s+(\ell-2)$, and set $C=C_0\cup C_1\cup C_2$.
		We then construct $3$ sets of voters $N_0$, $N_1$, $N_2$, so that $N_0=\{1, \dots, s\}$, 
		$|N_1|=\ell(s-1)+\ell$, $|N_2|=s\ell+\ell-3s+(\ell-2)$ (note that $|N_2|\ge (\ell-1)$ as we assume that $\ell\ge 3$).
		For each $i\in N_0$ we set $A_i=\{u_j\mid (u_j, v_i)\in E\}$, 
		and for each $i\in N_1$ we set $A_i=C_0\cup C_1$. The candidates in $C_2$
		are matched to voters in $N_2$: each voter in $N_2$ approves exactly one candidate in $C_2$, 
		and each candidate in $C_2$ is approved by exactly one voter in $N_2$.
		Denote the resulting list of ballots by $\calA$. 
		Finally, we set $k=2\ell-2$, and let $W=C_1\cup X$, where $X$ is a subset of $C_2$ with $|X|=\ell-1$.
		Note that the number of voters $n$ is given by $s+\ell(s-1)+\ell+s\ell+\ell-3s+(\ell-2)=2(s+1)(\ell-1)$, so $\frac{n}{k}=s+1$.
		
		Suppose first that we started with a ``yes''-instance of {\sc Balanced Biclique}, and let $(L', R')$
		be the respective $\ell$-by-$\ell$ biclique. Let $C^*=L'$ and $N^*=R'\cup N_1$.
		Then $|N^*|=\ell (s+1)=\ell \frac{n}{k}$, all voters in $N^*$ approve all candidates in $C^*$, $|C^*|=\ell$,
		but all voters in $N^*$ together are only represented by $\ell-1$ candidates in $W$. Hence, $W$
		fails to provide $\ell$-proportional justified representation for $(\calA, k)$.

		Conversely, suppose that $W$ fails to provide \pjr for $(\calA, k)$. That is, there exists
		a value $j>0$, a set $N^*$ of $j(s+1)$ voters and a set $C^*$ of $j$ candidates so that all voters
		in $N^*$ approve of all candidates in $C^*$, but all voters in $N^*$ together are only represented by
		less than $j$ candidates in $W$. Note that, since $s>1$ and $j\ge 1$, we have $N^*\cap N_2=\emptyset$.
		Further, since $|N^*|=j(s+1)\ge s+1$ and $|N_0|=s$, it follows that $N^*$ contains one voter from $N_1$.
		So, all voters in $N^*$ together are represented by exactly $\ell-1$ candidates in $W$. This implies that $j\ge \ell$. 
		As $N^*=j(s+1)\ge \ell(s+1)$, it follows that $|N^*\cap N_0|\ge \ell$.
		Since $N^*$ contains voters from both $N_0$ and $N_1$, it follows that $C^*\subseteq C_0$.
		Thus, there are at least $\ell$ voters in $N^*\cap N_0$ who approve the same $j\ge \ell$ candidates
		in $C_0$; any set of $\ell$ such voters and $\ell$ such candidates corresponds to an $\ell$-by-$\ell$ biclique
		in the input graph.
		\end{proof}

Note that although there is a polynomial-time algorithm to compute a committee that achieves \pjr~\citep{BFJL16a,SFF16a}, we have proved that checking whether any arbitrary committee achieves \pjr is coNP-complete. 
We complement the negative computational result by showing that testing \pjr is computationally tractable if one of the following parameters is bounded. 
		
		\begin{itemize}
			\item $m=|C|$
			\item $n=|N|$
				\item $a=\max_{i\in N}|A_i|$  (maximum size of approval sets).
				\item $d=\max_{c\in C}|\{i\in N\midd c\in A_i\}|$ (maximum number of approvals of a candidate). 
				\end{itemize}


We first observe that testing \pjr  is in FPT with parameter $n$.

\begin{theorem}\label{th:fpt-n}
Testing \pjr  is in FPT with parameter $n$ and takes time at most $O(2^{n}{mn})$.
	\end{theorem}
\begin{proof}
	Suppose we want to check whether $W\subset C$ satisfies \pjr.
If $n$ is bounded then one can simply brute force all the possible violating sets $X\subseteq N$ of voters and check that if $|X|\geq \ell\frac{n}{k} \text{ and } |\cap_{i\in X}A_i|\geq \ell$ then it must be that  $|W\cap (\cup_{i\in X}A_i)|\geq \ell$.
\end{proof}

Next, we prove that testing \pjr  is in FPT with parameter $m$.

\begin{theorem}\label{th:fpt-m}
	Testing \pjr  is in FPT with parameter $m$ and takes time at most $O(2^{m}{m}^3n)$.
	\end{theorem}
	\begin{proof}
		Suppose we want to check whether $W\subset C$ satisfies \pjr.
		Note that it is sufficient to show that testing $\ell$-\pjr is in FPT with parameter $m$. Note that $\ell\leq |X|k/n$ for all $X\subset N$. Since $|X|\leq n$, $\ell\leq k\leq m$. 
		
We go through all the subsets $S\in 2^C$ of size $\ell$. 
Each set $S$ is viewed as the intersection of possible objecting/deviating set of voters.
 For each $S$, we find the corresponding set of voters $X_S$ as follows: 
  
\[X_S=\{i\in N\midd A_i\subseteq S\}.\]
 
We return no (i.e., $W$ does not satisfy $\ell$-\pjr) if $|X_S|\geq \ell\frac{n}{k}$, $|\cap_{i\in X}A_i|\geq \ell$ but $|W\cap (\cup_{i\in X}A_i)|< \ell$.
		If we do not return no for any $S$, in that case we return yes (i.e., $W$ satisfies $\ell$-\pjr).

We now argue that it takes at most $O(2^{m}{m}^3n)$ operations to test \pjr. To check $\ell$-\pjr, we go through $2^{m}$ sets. For each set $S$, we find $X_S$ which takes ${m}^2n$ steps. After that we find $\cup_{i\in X}A_i$ which takes an additional $mn$ operations. Hence it takes $O(2^{m}{m}^2n)$ operations to test $\ell$-\pjr and it takes $O(2^{m}{m}^3n)$ operations to test \pjr.
\end{proof}

If $a=\max_{i\in N}|A_i|$ is bounded, then \pjr can be tested in polynomial time. 

\begin{theorem}\label{th:p-a}
	If $a$ is bounded, testing \pjr is solvable in polynomial time $O(m^{a+1}{m}^2n)$.
	\end{theorem}
	\begin{proof}
		Suppose we want to check whether $W\subset C$ satisfies \pjr.
		Note that it is sufficient to show that testing $\ell$-\pjr is polynomial-time solvable for each $\ell$ if $a$ is bounded. Note that we only need to consider $\ell\leq a$ because the maximum size of intersection of any set of approval sets is at most $a$ which means that $|\cap_{i\in X}A_i|\leq a$. For $\ell$ larger than $a$, $\ell$-\pjr is trivially satisfied. 

We now describe the algorithm to test $\ell$-\pjr for all $\ell\leq a$.
We go through all the subsets $S\in 2^C$ of size $\ell\leq a$.
There are at most ${m \choose \ell}=\frac{m!}{(m-\ell)!(\ell)!}$ such sets. Since $\ell\leq a$ and $a$ is bounded, it implies that $a$ is constant as well and hence there at most $m^a$ different subsets to be considered. 

Each set $S$ is viewed as the intersection of possible objecting /deviating set of voters.
 For each $S$, we find the corresponding set of voters $X_S$ as follows:

\[X_S=\{i\in N\midd A_i\subseteq S\}.\]

We return no (i.e., $W$ does not satisfy $\ell$-\pjr) if $|X_S|\geq \ell\frac{n}{k}$, $|\cap_{i\in X}A_i|\geq \ell$ but $|W\cap (\cup_{i\in X}A_i)|< \ell$.
		If we do not return no for any $S$, in that case we return yes (i.e., $W$ satisfies $\ell$-\pjr).

We now argue that it takes at most $O(m^a{m}^3n)$ operations to test \pjr. 
To check $\ell$-\pjr, we go through at most $m^a$ sets. For each set $S$, we find $X_S$ which takes ${m}^2n$ steps. After that we find $\cup_{i\in X}A_i$ which takes an additional $mn$ operations. Hence it takes $O(m^a{m}^2n)$ operations to test $\ell$-\pjr and it takes $O(m^{a+1}{m}^2n)$ operations to test \pjr.
\end{proof}

Finally, we show that if $d=\max_{c\in C}|\{i\in N\midd c\in A_i\}|$ is bounded, then \pjr can be tested in polynomial time.

\begin{theorem}\label{th:p-d}
	If $d$ is bounded, testing \pjr is solvable in polynomial time $O(n^ddnm^2)$.
	\end{theorem}
	\begin{proof}
		Suppose we want to check whether $W\subset C$ satisfies \pjr.
		Note that for any deviating/objecting set of voters $X$, $|\cap_{i\in X}A_i|\leq d$. The reason is that any candidate $c$ is approved by at most $d$ voters. Hence in order to check for a violation of $\ell$-\pjr, we just need to check for sets of voters of size at most $d$. There are at most ${n \choose d}\leq n^d$ such set of voters. For each such coalition of voters, we just need to check for $\ell$-\pjr for $\ell=1,\ldots, d/k$. So $\ell\leq d$. For each $X$ among the at most $n^d$ sets of voters, we need to test for $\ell$-\pjr which requires us to compute 
	$|\cap_{i\in X}A_i|$, $|W\cap (\cup_{i\in X}A_i)|$, and $|X|$. Hence testing $\ell$-\pjr of $W$ takes 	$O(n^dnm^2)$ time and testing \pjr takes time $O(n^ddnm^2)$.
\end{proof}

In this paper, we examined the complexity of testing \pjr, an interesting new axiom in committee voting. The arguments for all of our positive algorithmic results also hold for testing \ejr rather than \pjr. It will be interesting to see whether testing \pjr is FPT with respect to parameters $a$ or $d$.


%
%


	

\normalsize

%



 \end{document}